\documentclass[12pt,leqno,letterpaper]{article}

\usepackage{amsmath,amsthm,enumerate,amssymb}
\usepackage[latin1]{inputenc}
\usepackage[T1]{fontenc}
\usepackage[english]{babel}
\usepackage{graphicx}

\newtheorem{theorem}{Theorem}[section]
\newtheorem{proposition}[theorem]{Proposition}
\newtheorem{lemma}[theorem]{Lemma}

\newtheorem{definition}[theorem]{Definition}

\newcommand{\tr}{{\rm Tr\hskip -0.2em}~}

\begin{document}

\title{Metric adjusted skew information: Convexity and restricted forms of superadditivity}
\author{Liang Cai and Frank Hansen}
\date{July 30, 2009\\
{\tiny Revised August 24, 2009}}

\maketitle

\begin{abstract}

We give a truly elementary proof of the convexity of metric adjusted skew
information following an idea of Effros. We extend earlier results of weak forms of superadditivity to general metric adjusted skew informations. Recently, Luo and Zhang introduced the notion of semi-quantum states on a bipartite system and proved superadditivity of the Wigner-Yanase-Dyson skew informations for such states. We extend this result to general metric adjusted skew informations. We finally show that a
recently introduced extension to parameter values $ 1<p\le 2 $ of the WYD-information is a special case of (unbounded) metric adjusted skew information.
\end{abstract}

\section{Introduction}

The theory of measures of quantum information was initiated by Wigner and Yanase \cite{kn:wigner:1963} who wanted to find a good measure of our knowledge of a difficult-to-measure observable (state)
with respect to a conserved quantity. They were motivated by earlier observations \cite{kn:wigner:1952, kn:araki:1960} indicating that the obtainable accuracy of the measurement of a physical observable, in the presence of a conservation law, is limited if the operators representing the observable and the conserved quantity do not commute. Wigner and Yanase discussed a number of postulates that such a measure should satisfy and proposed, tentatively, the so-called {\it skew information} defined by
\[
I_{\rho}(A)=-\frac{1}{2}{\rm Tr}([\rho^{\frac{1}{2}},A]^2),
\]
where $\rho$ is a state (density matrix) and $A$ is an observable (self-adjoint matrix). The most important
requirement to a measure of quantum information is that knowledge decreases under the mixing of states; or equivalently that the measure is convex in the state variable. They proved this property for the skew information, but Dyson suggested that the measures defined by
\begin{equation}\label{WYD information}
I_{\rho}(p,A)=-\frac{1}{2}{\rm Tr}([\rho^p,A] \cdot [\rho^{1-p},A]),
\end{equation}
where $ p $ is a parameter $ (0<p<1) $ may have the same property.
This became the celebrated Wigner-Yanase-Dyson conjecture later
proved by Lieb \cite{Lieb73}. In addition to the convexity
requirement Wigner and Yanase suggested that the measure should be
additive with respect to the aggregation of isolated subsystems and,
for an isolated system, independent of time. These requirements are
easily seen to be satisfied for both the Wigner-Yanase skew
information and for the Dyson generalization in (\ref{WYD
information}).

Finally, Wigner and Yanase conjectured that the skew information should satisfy a certain superadditivity condition coming from thermodynamics, where it is satisfied for both classical and quantum mechanical systems. It reflects the loss of information about statistical correlations between two
subsystems when they are only considered separately. Wigner and Yanase
conjectured that the skew information also possesses this property and
proved it when the state of the aggregated system is pure. The conjecture was
generally believed to be true \cite{kn:luo:2007} until the second author gave a counter example in~\cite{kn:hansen:2007:2}.

The first author, Li and Luo \cite{kn:cai:2008} discussed weak forms
of superadditivity for the Wigner-Yanase information, and Luo and
Zhang \cite{Luo072,LuoZhang08} introduced the notion of semi-quantum
states on a bipartite system and proved superadditivity for the
Wigner-Yanase-Dyson skew informations for such states. We continue
these inquiries and extend the obtained results to general metric
adjusted skew informations.

\subsection{Monotone metrics}

 \begin{definition} We denote by $ {\cal F}_{op} $ the set of functions $ f\colon\mathbb R_{+}\to\mathbb R_{+} $ such that
    \begin{enumerate}[(i)]
    \item $ f $ is operator monotone,
    \item $ f(t)=tf(t^{-1}) $ for all $ t>0, $
    \item $ f(1)=1. $
    \end{enumerate}
    \end{definition}

 Since $ f $ is increasing, it may be extended to a continuous function defined on $ [0,\infty) $ with $ f(0)\ge 0. $ We say that $ f $ is regular if $ f(0)>0 $ and non-regular if $ f(0)=0. $
 A Morozova-Censov function $ c $ is a function  of two positive variables
 given on the form
\[
c(x,y)=\frac{1}{y f(xy^{-1})}\qquad x,y>0
 \]
where $ f\in {\cal F}_{op}. $ It defines a Rimannian metric $ K_\rho^c $ on the tangent space of the state manifold (the space of all positive definite density matrices supported by the underlying Hilbert space), and it is given by
\[
K_\rho^c(A,B)=\tr A^* c(L_\rho,R_\rho) B
\]
where $ L_\rho $ and $ R_\rho $ are the positive definite commuting left and right multiplication operators by $ \rho. $
The metric is customarily extended to all linear operators $ A $ and $ B $ in the underlying Hilbert space. It is decreasing in the sense that
\[
K^c_{T(\rho)}(T(A),T(A))\leq K^c_\rho(A, A)
\]
for any self-adjoint A and completely positive trace preserving map T, cf. \cite{kn:petz:1996:2,kn:ruskai:1999}.

The combined efforts of Cencov and Petz established that any
decreasing Rimannian metric on the state manifold is given in this
way \cite{kn:censov:1982,kn:morozova:1990,kn:petz:1996:2}. We say
that the metric is regular if the generating operator monotone
function is regular. Recently, a correspondence
between the regular and the non-regular operator monotone functions in $ {\cal F}_{op} $ has been
constructed by Gibilisco, the second author and Isola \cite{kn:gibilisco:2008}.

\subsection{Metric adjusted skew information}

Let $ c $ be the Morozova-Cencov function of a regular decreasing metric. The corresponding {\it metric adjusted skew information} is defined by
\begin{equation}\label{metric adjusted skew information}
I^c_\rho(A) = \frac{m(c)}{2} K^c_\rho(i[\rho, A], i[\rho, A]),
\end{equation}
where $ m(c)=f(0) $ is the metric constant.
It may be extended \cite[Theorem 3.8]{kn:hansen:2008:1} to positive semi-definite states $ \rho $ and
arbitrary operators $ A. $ The following theorem is taken from
\cite{kn:hansen:2008:1}.

    \begin{theorem}\label{th:hansen:1}
    Let $ c $ be a regular Morozova-Chentsov function, and consider the finite-dimensional Hilbert space $ \mathbf C^n $ for some natural number $ n. $
    \begin{enumerate}
    \item The metric adjusted skew information is a convex function,
    $ \rho\to I^c_\rho(A), $ on the manifold of states for any $ n\times n $ matrix $ A. $

    \item For $ \rho=\rho_1\otimes\rho_2 $ and $ A=A_1\otimes 1 + 1\otimes A_2 $ we have
    \[
    I^c_\rho(A)=I^c_{\rho_1}(A_1)+I^c_{\rho_2}(A_2).
    \]

    \item If $ A  $ commutes with an Hamiltonian operator $ H $ then
    \[
    I^c_{\rho_t}(A)=I^c_\rho(A)\qquad t\ge 0,
    \]
    where $ \rho_t=e^{itH}\rho e^{-itH}. $

    \item For any pure state $ \rho $ (one-dimensional projection) we have
    \[
    I^c_\rho(A)={\operatorname{Var}}_\rho(A)
    \]
    for any $ n\times n $ matrix $ A. $

    \item For any density matrix $ \rho $ and $ n\times n $ matrix $ A $ we have
    \[
    0\le I^c_\rho(A)\le{\operatorname{Var}}_\rho(A).
    \]
  \end{enumerate}
    \end{theorem}

     The first item shows that the metric adjusted
    skew information is decreasing under the mixing of states. The second shows that it is additive with respect to the aggregation of
    isolated subsystems, and the third that it, for an isolated system, is independent of time. These requirements were considered essential
    by Wigner and Yanase \cite{kn:wigner:1963} to an effective measure of quantum information.

    Notice that the metric adjusted skew information gives no information about observables commuting with the state variable.
    It would thus appear that the passage from a regular decreasing metric to the metric adjusted skew information is irreversible.

    Luo \cite{kn:luo:2003,kn:luo:2003:2} was the first to notice that the Wigner-Yanase information is bounded by the variance. Recent developments can be found in \cite{kn:andai:2008,kn:audenaert:2008,kn:gibilisco:2007:2}.

Hasegawa and Petz \cite[Theorem 2]{kn:hasegawa:1996} proved that the functions $ f_p $ defined by setting
\begin{equation}\label{family of operator monotone functions}
f_p(t)=\left\{
\begin{array}{ll}
\displaystyle p(1-p)\frac{(t-1)^2}{(t^p-1)(t^{1-p}-1)}\qquad &t>0,\, t\ne 1\\[1ex]
1 &t=1
\end{array}\right.
\end{equation}
are operator monotone for $ 0<p<1; $ they are in fact regular functions in $ \mathcal{F}_{op}$
with $ \lim_{t\to 0}f_p(t)=p(1-p). $ The corresponding Morozova-Chentsov functions are given by
\[
c(x,y)=\frac{1}{p(1-p)}\cdot\frac{(x^p-y^p)(x^{1-p}-y^{1-p})}{(x-y)^2}\qquad x,y>0,
\]
and the associated metric adjusted skew information is therefore the Wigner-Yanase-Dyson skew information
\begin{equation}\label{WYD skew information}
I_\rho^c(A)=-\frac{1}{2}\tr [\rho^p,A]\cdot [\rho^{1-p},A].
\end{equation}
The convexity in the state variable of the metric adjusted skew information is thus a generalization of the convexity of the Wigner-Yanase-Dyson information proved by Lieb \cite{Lieb73}.

Jen\v{c}ova and Ruskai considered in the paper \cite{kn:jencova:2009} the extension of the Wigner-Yanase-Dyson information to parameter values $ 1< p\le 2, $ and they proved that also the extension is convex in the state variable. We show that this extension is associated with monotone metrics, and if we relax the definition of a metric adjusted skew information to allow for non-regular metrics,  Jen\v{c}ova and Ruskai's extension can be understood as an (unbounded) metric adjusted skew information. In particular, the extension automatically enjoys all the well-known properties of a metric adjusted skew information except boundedness.

\section{Convexity of skew information}

Convexity of the metric adjusted skew information in the state
variable was proved in \cite{kn:hansen:2008:1}, but the given proof
relies on Löwner's deep theory of the relationship between operator
monotone functions and the theory of analytic functions. Recently,
Effros \cite[Theorem 2.2]{Effros} put forward a simple idea\footnote{There is some similarity to the approach used in \cite{kn:petz:1986}.} that
made it possible to prove all the major quantum information
inequalities by reduction to simple block matrix manipulations. The
same scheme may also be used to prove convexity in the state
variable of the metric adjusted skew information.

\begin{theorem}[Effros]\label{effros transform}
Suppose that $h$ is operator convex. When restricted to positive commuting matrices the function $ g $ defined by
\begin{equation}
g(L, R)= h \left(\frac{L}{R}\right) R
\end{equation}
is jointly convex, that is if the commutator $ [L,R]=0 $ and
 \[
 L = \lambda L_1+(1-\lambda)L_2
 \quad\text{and}\quad
R=\lambda R_1+(1-\lambda)R_2
\]
where also the commutators $ [L_1, R_1]=0 $ and $ [L_2, R_2]=0, $ then
\[
g(L,R)\leq \lambda g(L_1, R_1)+(1-\lambda)g(L_2, R_2)
\]
for $0\leq \lambda\leq 1.$
\end{theorem}

By a simple reduction the metric adjusted skew information may be written on the form
\[
\label{representation of metric adjusted skew information}
I^c_\rho(A)=\frac{m(c)}{2}\displaystyle\tr A\, \hat c(L_\rho,R_\rho)A,
\]
where
\begin{equation}\label{c hat}
\hat c(x,y)=(x-y)^2 c(x,y)\qquad x,y>0.
\end{equation}
It is easy to prove that the metric adjusted skew information is convex in the state variable $ \rho $ if $ \hat c $ is operator convex, see for example \cite[Theorem 1.1]{kn:hansen:2006:3}. The key observation is now to realize that
$ \hat c $ is the Effros transform in Theorem~\ref{effros transform} of the function
\[
h(t)=\frac{(t-1)^2}{f(t)}\qquad t>0.
\]
Operator convexity of $ \hat c $ and convexity of the metric adjusted skew information will thus follow if just $ h $ is operator monotone. It turns out that a proof of this assertion can be accomplished by the same simple matrix manipulations that underlies Effros's  proof of the quantum mechanical inequalities.

\begin{theorem}\label{theorem: operator convex functions}
Let $f\colon\mathbf R_+\to \mathbf R_+ $ be an operator monotone function. Then the function
\[
h(t)=\frac{(t-1)^2}{f(t)}\qquad t>0
\]
is operator convex\footnote{Notice that the theorem cannot be inverted. There are
operator convex functions $ h\colon\mathbf R_+\to\mathbf R_+ $ such that the function $ f(t)=(t-1)^2/h(t) $ is not operator monotone.}.
\end{theorem}

\begin{proof} We first assume that the left limit
\[
f(0)=\lim_{t\to 0} f(t)>0
\]
and write $ h $ as a sum $ h=h_1+h_2+h_3 $ of the three functions
\[
h_1(t)=\frac{t^2}{f(t)}\,,\quad h_2(t)=\frac{-2t}{f(t)}\,,\quad h_3(t)=\frac{1}{f(t)}
\]
defined on the positive half-axis. We prove the statement of the theorem by showing that each of these three functions is operator convex, and we shall do so by elementary methods without the usage of Löwner's theorem.

1. The function $ f $ is strictly positive, defined on the positive half-axis, and operator monotone, hence also the function $ t\to t f(t)^{-1} $ is operator monotone \cite[2.6 Corollary]{HansenPedersen82}. The function $ h_1 $ may be extended to a continuous function defined on $ [0,\infty) $ with $ h_1(0)=0, $
and since the function
\[
\frac{h_1(t)}{t}=\frac{t}{f(t)}
\]
is operator monotone, it follows \cite[2.4 Theorem]{HansenPedersen82} that $ h_1 $ is operator convex.

2. It follows from the remarks above that the function $ t\to tf(t)^{-1} $ may be extended to a non-negative operator monotone function defined on $ [0,\infty). $
It is consequently also operator concave \cite[2.5 Theorem]{HansenPedersen82}, hence $ h_2 $ is operator convex.

3. The function $ f $ itself is
operator concave \cite[2.5 Theorem]{HansenPedersen82} and since inversion is operator decreasing and operator convex we obtain
\[
\begin{array}{rl}
h_3(\lambda x +(1-\lambda)y)&\displaystyle=
\frac{1}{f(\lambda x +(1-\lambda)y)}\\[3ex]
&\displaystyle\leq\frac{1}{\lambda f(x)+(1-\lambda)f(y)}\\[3ex]
&\displaystyle\leq \lambda \frac{1}{f(x)}+(1-\lambda)\frac{1}{f(y)}\\[3ex]
&\displaystyle=\lambda h_3(x)+(1-\lambda)h_3(y)
\end{array}
\]
for $0\leq \lambda \leq 1$ and positive semi-definite operators $x $ and $y.$
From these three observations we obtain the desired conclusion.

In the general case we first apply the obtained result to the functions $ f_\varepsilon(t)=f(t+\varepsilon) $ for $ \varepsilon>0. $ By letting $ \varepsilon\to 0, $ we thus obtain $ h $ as a point-wise limit of operator convex functions, hence $ h $ is also operator convex.
\end{proof}

Notice that the above theorem is of a general nature and does not require $ f $ to be a function in $ {\cal F}_{op}\,. $ In conclusion, we have proved convexity of the metric adjusted skew information by completely elementary methods.

\section{Weak forms of superadditivity}

Lieb \cite{Lieb73} noticed that the Wigner-Yanase-Dyson information is superadditive in the special case where the conserved quantity is of the form $ A\otimes 1_2 $ and we may effortlessly extend this result to any metric adjusted skew information by using that the underlying metric is decreasing.

\begin{lemma}\label{lieb lemma}
Let $f$ be a regular function in $ \mathcal{F}_{op}$ and let $ c $ be the associated Morozova-Chentsov function. Let furthermore $\rho$ be a bipartite density operator on a tensor product $ H_1\otimes H_2 $ of two parties.
If $ A $ is an observable of the first party then
\begin{align}\label{coroweaksup}
I_\rho^c(A\otimes {\mathbf 1_2}) \geq I^c_{\rho_1}(A),
\end{align}
where $ \rho_1=\tr_2\, \rho $ is the partial trace of $ \rho $ on $ H_1. $
\end{lemma}

\begin{proof}
We first notice that $ \tr_2 \rho(A\otimes 1_2)=\rho_1 A. $ Indeed, for any vectors $ \xi,\eta\in H_1 $ and orthonormal basis $ (e_i)_{i\in I} $ in $ H_2 $ we have
\[
\begin{array}{rl}
(\xi\mid\tr_2\rho(A\otimes\mathbf 1_2)\eta)
&\displaystyle=\sum_{i\in I}(\xi\otimes e_i\mid \rho(A\otimes\mathbf 1_2)(\eta\otimes e_i))\\[3ex]
&=\displaystyle\sum_{i\in I}(\xi\otimes e_i\mid \rho(A\eta\otimes e_i))\\[3ex]
&=(\xi\mid \rho_1A\eta).
\end{array}
\]
It follows that $ \tr_2 i[\rho, A\otimes\mathbf 1_2]=i[\rho_1,A] $ for any self-adjoint $ A $ on $ H_1. $ Since the partial trace is completely positive and trace preserving and the metric $ K_\rho^c $ is decreasing we obtain
\[
\begin{array}{rl}
I^c_\rho(A\otimes\mathbf 1_2)&=\displaystyle\frac{m(c)}{2}
K_{\rho}^c(i[\rho,A\otimes {\bf 1_2}],i[\rho,A\otimes\mathbf 1_2])\\[2ex]
&\geq\displaystyle\frac{m(c)}{2} K_{\rho_1}^c(\tr_2 i[\rho,A\otimes\mathbf 1_2], \tr_2
i[\rho,A\otimes\mathbf 1_2])\\[2ex]
&=\displaystyle\frac{m(c)}{2}K_{\rho_1}^c(i[\rho_1,A],i[\rho_1,A])=I^c_{\rho_1}(A)
\end{array}
\]
and the proof is complete.
\end{proof}

In the same spirit, we extend two weak forms of superadditivity for the Wigner-Yanase information \cite{kn:cai:2008} to general metric adjusted skew informations.

\begin{proposition}
Let $f$ be a regular function in $ \mathcal{F}_{op}$ and let $ c $ be the associated Morozova-Chentsov function. Let furthermore $\rho$ be a bipartite density operator on a tensor product $ H_1\otimes H_2 $ of two parties,
then we obtain the following two weak forms of superadditivity for the metric adjusted skew information:

\begin{equation}\label{weaksup}
I^c_\rho(A\otimes {\mathbf 1_2}+{\mathbf 1_1}\otimes B)
\geq\frac{1}{2}\big(I^c_{\rho_1}(A)+I^c_{\rho_2}(B)\big)
\end{equation}
and
\begin{equation}
\begin{array}{l}
I^c_\rho(A\otimes {\mathbf 1_2}+{\mathbf 1_1}\otimes B)+
I^c_\rho(A\otimes {\mathbf 1_2}-{\mathbf 1_1}\otimes B)\\[1.5ex]
\hskip 9.9em\geq 2\big(I^c_{\rho_1}(A)+I^c_{\rho_2}(B)\big),
\end{array}
\end{equation}
where $ A $ is an observable of the first party and $ B $ is an observable of the second party.
\end{proposition}

\begin{proof}
We begin by noticing that the partial trace of a commutator
\[
\tr_2 [\rho, \mathbf 1_1\otimes B]=0.
\]
Indeed, if $ \rho=\sigma_1\otimes\sigma_2 $ is a simple tensor, then a simple calculation yields
\[
(\xi\mid (\tr_2 [\sigma_1\otimes\sigma_2, \mathbf 1_1\otimes B])\eta]=
(\xi\mid\sigma_1\eta) \tr[\sigma_2,B]=0
\]
for vectors $ \xi,\eta\in H_1, $ and since a bipartite state $ \rho $ is a sum of simple tensors, the statement follows. We thus obtain
\[
\tr_2 [\rho, A\otimes \mathbf 1_2+\mathbf 1_1\otimes B]=
\tr_2 [\rho, A\otimes \mathbf 1_2]=[\rho_1, A],
\]
where $ \rho_1=\tr_2\rho $ is the partial trace of $ \rho $ on the first party.
By using that the metric $ K^c_\rho $ is decreasing and the partial trace is completely positive and trace preserving, we thus obtain
\[
\begin{array}{rl}
I^c_\rho(A\otimes\mathbf 1_2+\mathbf 1_1\otimes B)
&=
\displaystyle\frac{m(c)}{2} K^c_\rho(i[\rho,X^+], i[\rho,X^+])\\[2ex]
&\ge\displaystyle\frac{m(c)}{2} K^c_{\rho_1}(i[\rho_1,A], i[\rho_1,A])\\[2ex]
&=I^c_{\rho_1}(A),
\end{array}
\]
where $ X^+=A\otimes\mathbf 1_2+\mathbf 1_1\otimes B. $
By symmetry, we similarly obtain
\[
I^c_\rho(A\otimes\mathbf 1_2+\mathbf 1_1\otimes B)\ge I^c_{\rho_2}(B),
\]
and the first assertion follows.

The second statement is an easy consequence of the parallelogram identity
\[
I^c_\rho(X^+)+I^c_\rho(X^-) =2\big(I^c_\rho(A\otimes {\mathbf
1_2})+I^c_\rho({\mathbf 1_1}\otimes B)\big),
\]
where in addition $ X^-=A\otimes\mathbf 1_2-\mathbf 1_1\otimes B, $ and the inequality in Lemma~\ref{lieb lemma}.
\end{proof}

\section{Semi-quantum states}

A local von Neumann measurement $ P $ of the first party of a bipartite state $ \rho $ on a tensor product $ H_1\otimes H_2 $ is given by
\[
P(\rho)=\sum_{i\in I}(P_i\otimes \mathbf 1_2)\rho (P_i\otimes \mathbf 1_2)
\]
where $ \{P_i\}_{i\in I} $ is a resolution of the identity on $ H_1 $ in one-dimensional projections. We may similarly define a local von Neumann measurement of the second party.
The concept of semi-quantum states was introduced by Luo and Zhang \cite{LuoZhang08}.

\begin{definition}
 A bipartite state $\rho$ is called a semi-quantum state if there exists a local von Neumann measurement $P=\{P_i\}_{i\in I}$ of the first (or second) party leaving $\rho$ invariant, i.e. $ P(\rho)=\rho. $
\end{definition}

Luo and Zhang \cite{LuoZhang08} showed that $\rho$ is a semi-quantum state
(corresponding to the local von Neumann measurement $P) $ if and only if
\begin{equation}\label{representation of semi-quantum state}
\rho=\sum_{i\in I} p_i P_i\otimes\rho_i
\end{equation}
where $ (p_i)_{i\in I} $ is a probability distribution and $ \rho_i $ for each
$ i\in I $ is a state of the second party. They also proved that the Wigner-Yanase-Dyson skew information is superadditive when the bipartite state is semi-quantum. We extend this result to general metric adjusted skew informations.

\begin{theorem}
Let $f$ be a regular function in $ \mathcal{F}_{op}$ and let $ c $ be the associated Morozova-Chentsov function. Let furthermore $\rho$ be a semi-quantum state on a tensor product $ H_1\otimes H_2 $ of two parties,
then we obtain superadditivity of the metric adjusted skew information
\[
I^c_\rho(A \otimes\mathbf 1_2+ \mathbf 1_1\otimes B) \geq I^c_{\rho_1}(A)+ I^c_{\rho_2}(B),
\]
where $ A $ is an observable of the first party and $ B $ is an observable of the second party.

\end{theorem}

\begin{proof}
The essential step is to prove that
\begin{equation}\label{vanishing cross terms}
K_{\rho}^c(i[\rho, A\otimes\mathbf 1_2], i[\rho, \mathbf 1_1\otimes B])=0
\end{equation}
for a semi-quantum state $ \rho. $ We first notice that
\[
K_{\rho}^c(i[\rho , A\otimes\mathbf 1_2], i[\rho, \mathbf 1_1\otimes B])=
\tr (A\otimes\mathbf 1_2) \hat c(L_\rho,R_\rho) (\mathbf 1_1\otimes B)
\]
where $ \hat c $ is defined in (\ref{c hat}). We write $ \rho $ on the form
(\ref{representation of semi-quantum state}) and choose for each $ i\in I $
a spectral resolution
\[
\rho_i=\sum_{j\in J} \lambda_{ij} Q_{ij}
\]
in terms of one-dimensional projections and obtain in this way a spectral resolution
\[
\rho=\sum_{i\in I;\,j\in J}p_i\lambda_{ij} (P_i\otimes Q_{ij})
\]
of the semi-quantum state $ \rho. $ We may then calculate
\[
\begin{array}{l}
\hat c(L_\rho,R_\rho)(\mathbf 1_1\otimes B)\\[3ex]
=\displaystyle\sum_{i,i'\in I;\,j,j'\in J} \hat c(p_i\lambda_{ij}, p_{i'}\lambda_{i'j'})
(P_i\otimes Q_{ij})(\mathbf 1_1\otimes B)(P_{i'}\otimes Q_{i'j'})\\[4ex]
=\displaystyle\sum_{i\in I;\,j,j'\in J} \hat c(p_i\lambda_{ij}, p_i\lambda_{ij'})
(P_i\otimes Q_{ij} B Q_{ij'})
\end{array}
\]
and obtain
\[
\tr (A\otimes\mathbf 1_2) \hat c(L_\rho,R_\rho) (\mathbf 1_1\otimes B)
=\sum_{i\in I;\,j\in J} \hat c(p_i\lambda_{ij}, p_i\lambda_{ij})
(\tr AP_i)(\tr BQ_{ij}).
\]
Since $ \hat c(x,x)=0 $ for $ x>0 $ we derive (\ref{vanishing cross terms}) as desired. The metric $ K^c_\rho $ is sesqui-linear, hence
\[
\begin{array}{l}
I^c_\rho(A \otimes\mathbf 1_2+ \mathbf 1_1\otimes B)\\[2ex]
=I^c_\rho(A \otimes\mathbf 1_2)+I^c_\rho(\mathbf 1_1\otimes B)+
m(c) K_{\rho}^c(i[\rho,A \otimes\mathbf 1_2], i[\rho, \mathbf 1_1\otimes
B])\\[2ex]
=I^c_\rho(A \otimes\mathbf 1_2)+I^c_\rho(\mathbf 1_1\otimes B)\\[2ex]
\ge I^c_{\rho_1}(A)+ I^c_{\rho_2}(B),
\end{array}
\]
where we used (\ref{vanishing cross terms}) and Lemma~\ref{lieb lemma}.
\end{proof}

\section{Unbounded skew information}

The definition in (\ref{metric adjusted skew information}) of the metric adjusted skew information requires the associated monotone metric to be regular. The monotone metric may thus be extended radially from the state manifold to the state space, and the metric adjusted skew information is then normalized such that it coincides with the variance on pure states. We may also introduce a metric adjusted skew information associated with a non-regular monotone metric by setting
\begin{equation}\label{unbounded metric adjusted skew information}
I^c_\rho(A) = K^c_\rho(i[\rho, A], i[\rho, A])
\end{equation}
if $ c $ is a non-regular Morozova-Chentsov function. This type of metric adjusted skew information is unbounded and can no longer be extended from the state manifold to the state space. However, it enjoys all the same general properties as a bounded metric adjusted skew information; it is non-negative and the convexity in the state variable follows directly from Theorem \ref{theorem: operator convex functions}.
Similarly, all the statements concerning limited forms of superadditivity appearing in this paper may be extended ad verbatim to unbounded metric adjusted skew informations.

\begin{theorem}
The functions $ f_p(t) $ defined in (\ref{family of operator monotone functions}) are for $ 1<p\le 2 $  non-regular functions in $ \mathcal{F}_{op}.$
\end{theorem}

\begin{proof}
The function $ f_p $ is for $ 1< p\le 2 $ positive since $ p(1-p)<0 $ and exactly one of the terms in the denominator is negative for each $ t\ne 1. $ We notice that $ f_p(t)\to 1 $ for $ t\to 1 $ and that $ f_p(t)\to 0 $ for $ t\to 0. $ We only need to prove that $ f_p $ is operator monotone, and to obtain this we consider the identity
\[
f_p(t)=-p(1-p)\frac{t-1}{g_p(t)-1}
\]
where the function
\[
g_p(t)=\left\{
\begin{array}{ll}
\displaystyle \frac{t^p-1}{t-1}+\frac{t^{1-p}-1}{t-1}\qquad &t>0,\, t\ne 1\\[2ex]
1 &t=1.
\end{array}
\right.
\]
Since both $ t^p $ and $ t^{1-p} $ are operator convex for $ 1<p\le 2, $ it follows from \cite[Theorem 3.2]{kn:bendat:1955} that $ g_p $ is operator monotone, and it is therefore also operator concave (notice that this conclusion does not require $ g_p $ to be positive). By appealing to Bendat and Sherman's theorem once more and taking inverse we conclude that $ f_p $ is operator monotone.
\end{proof}

The monotone metrics corresponding to the functions $ f_p $ for $ 1<p< 2 $ have not been studied in the literature\footnote{There is a claim in \cite{kn:hasegawa:1997} that the functions $ f_p $ are operator monotone also for $ 1<p<2. $ However, it does not seem to be widely noticed that there is an unrecoverable error in the proof. In equation $ (3.3') $ in the reference there is an integral representation of functions written on the form $ (cx+d)/(ax+b)^2, $ and the key argument is that these functions should be operator monotone decreasing for positive coefficients; or equivalently that the functions $ (ax+b)^2/(cx+d) $ should be operator monotone. To this effect there is a reference to a paper \cite{HansenPedersen82} co-authored by the second order. However, there is no such claim in \cite{HansenPedersen82} and the assertion is not true.
If it were true, and since a point-wise limit of operator monotone
functions is again operator monotone, the claim would also be true for
non-negative coefficients. Then, taking $ a=1, $ $ b=0, $ $ c=0, $ and $ d=1 $ we would
obtain that $ x^2 $ is operator monotone; but that is not the case.}. Hasegawa and Petz \cite[Theorem 2]{kn:hasegawa:1996} proved operator monotonicity for $ 0<p<1 $ by ingeniously constructing integral representations. However, this type of representations only work for $ 0<p<1. $ By symmetry we also obtain that $ f_p $ is operator monotone for $ -1<p<0. $ We also notice that
\[
f_p(t)\to \frac{t-1}{\log t}\qquad\text{for}\quad p\to 0\quad\text{or}\quad p\to 1,
\]
and this is the function generating the Kubo metric. Similarly,
\[
f_p(t)=\frac{2t}{t+1}\qquad\text{for}\quad p=-1\quad\text{or}\quad p=2,
\]
and this is the function generating the minimal monotone metric. Therefore, with these extensions, we obtain monotone metrics for $ -1\le p\le 2. $

For $ 1<p\le 2 $ the (unbounded) metric adjusted skew information associated with the non-regular functions
$ f_p\in\mathcal{F}_{op} $ is given by
\begin{equation}\label{extension of WYD}
I_\rho^c(A)=\frac{-1}{p(1-p)}\tr [\rho^p,A]\cdot[\rho^{1-p},A].
\end{equation}
In fact, the only difference from (\ref{WYD skew information}) is the constant in front of the trace. But this is exactly the extension of the Wigner-Yanase-Dyson skew information to parameter values $ 1<p\le 2 $ studied by Jenc\v{o}va and Ruskai in \cite{kn:jencova:2009}. The above result shows that
this extension of the WYD-information is also associated with monotone metrics and can be understood in terms of the notion of metric adjusted skew information. The main difference is that the metric is regular for
$ 0<p<1 $ but non-regular for $ 1<p\le 2. $

It is therefore immediate that the extension proposed by Jenc\v{o}va and Ruskai is non-negative and convex in the state variable. Furthermore, it satisfies all the restricted forms of monotonicity under partial traces studied in this paper.

{\small

      \bibliographystyle{plain}
    \bibliography{mathharv}

\vfill

      \noindent Liang Cai: Department of Mathematics, Beijing Institute of Technology, Beijing, China.

      \noindent Frank Hansen: Department of Economics, University
       of Copenhagen,     
       Øster Farimagsgade 5, building 26, DK-1353 Copenhagen K, Denmark.}

\end{document}